\documentclass{article}
\usepackage{hyperref}
\usepackage{
	amsmath,
	amssymb,
	amsthm,
	array,%Kayleigh
	caption,
	chngpage,%Kayleigh
	color,
	float,%Kayleigh
	graphicx,
	mathrsfs,
	multicol,%Kayleigh
	multirow,%Kayleigh
	rotating,%for sidewaystable
	sectsty,
	subcaption,
	tikz,
	verbatim,
	wrapfig
}

\usetikzlibrary{automata, positioning}
\usetikzlibrary{matrix,fit,backgrounds}
\newcommand{\restrict}{\upharpoonright}
\newcommand{\eps}{\varepsilon}
\newcommand{\abs}[1]{\lvert#1\rvert}
\newcommand{\Var}{\operatorname{Var}}
\newtheorem{thm}{Theorem}
\newtheorem{cor}[thm]{Corollary}

\newtheorem{con}[thm]{Conjecture}

\theoremstyle{definition}
\newtheorem{df}[thm]{Definition}
\theoremstyle{remark}
\newtheorem{rem}[thm]{Remark}
\newcommand{\E}{\mathbb E}

\begin{document}
		\title{
			Pricing complexity options
		}
		\author{
			Malihe Alikhani\\
			\small Department of Computer Science\\
			\small Rutgers, the State University of New Jersey\\
			\small Piscataway, NJ 08854
			\and
			Bj{\o}rn Kjos-Hanssen\thanks{
				Corresponding author:
				Bj\o rn Kjos-Hanssen,
				Department of Mathematics,
				University of Hawai\textquoteleft i at M\=anoa,
				Honolulu, HI 96822.
				Email: {\texttt bjoern.kjos-hanssen@hawaii.edu}.
				Tel.~no.~+1 (808) 956-8595.
				Fax~no.~+1 (808) 956-9139.
			}\\
			\small Department of Mathematics\\
			\small University of Hawai\textquoteleft i at M\=anoa\\
			\small Honolulu, HI 96822
			\and
			Amirarsalan Pakravan \\
			\small Department of Finance\\
			\small George Washington University\\
			\small Washington, DC 20052
			\and
			Babak Saadat\\
			\small Kash Co\\
			\small 625 N W Knoll Dr\\
			\small West Hollywood, CA 90069
		}
	\maketitle
	\begin{abstract}
		We consider options that pay the complexity deficiency
		of a sequence of up and down ticks of a stock upon exercise.
		We study the price of European and American versions of this option 
		numerically for automatic complexity, and
		theoretically for Kolmogorov complexity.
		We also consider
		run complexity, which is a restricted form of automatic complexity.

	\end{abstract}
	\tableofcontents
	\section{Introduction}
		In this article we consider the pricing of American and European options paying
		the \emph{complexity deficiency}, or intuitively the lack of complexity, of a sequence of up and down ticks for a financial security.
		The complexity notions we consider are
		plain and prefix-free Kolmogorov complexity, nondeterministic automatic complexity, and run complexity.
		\subsection{Motivation}
			%\paragraph{Why study complexity options?}
			We believe it may be of value in finance to have
			some notions of the complexity of a price path.
			Agents may want to insure against too complex or too simple price paths for a stock, for example.
			A very simple or complex path may be a sign that something is going on that the agent is not aware of.

			Weather is somewhat periodic, and automatic complexity measures
			periodicity, to some extent. Hence a complexity option may be used as a weather derivative.

			Casino owners may want to ensure that their casinos are truly random, so as to avoid unexpected losses.
			In general, anyone who makes an assumption of randomness may want to hedge that,
			as true randomness is not easy to guarantee, or even completely well-defined.

			\paragraph{Automatic complexity: between two extremes.} Of course, we can insure against certain types of non-randomness in simple ways.
			We can insure against a dramatic fall of a stock price by selling the stock short.
			This corresponds to \emph{run complexity} (Section \ref{Malihe}).
			At the other end, one cannot use \emph{Kolmogorov complexity} (Section \ref{sec:KC}) as a basis for the security,
			because Kolmogorov complexity is not computable.
			The \emph{nondeterministic automatic complexity},
			being both
			\begin{itemize}
				\item powerful enough to discern a variety of patterns, and at the same time
				\item single-exponential time computable,
			\end{itemize}
			may be a promising middle ground.
		\subsection{Automatic complexity and the idea of complexity deficiency}
			Kolmogorov complexity is an important notion that in a way is to complexity
			as Turing computability is to computability.
			It is computably approximable, but unfortunately not computable.
	 		As a remedy, \cite{MR1897300} defined the \emph{automatic complexity}
			of a finite binary string \(x=x_1\ldots x_n\) to be 
			the least number \(A_D(x)\) of states of a {deterministic finite automaton} \(M\) such that 
			\(x\) is the only string of length \(n\) in the language accepted by \(M\). 

			Automatic complexity is computable, but it does have a couple of awkward properties that make us want to tweak its definition.
			First, many of the automata used to witness the complexity have
			a \emph{dead state} whose sole purpose is to absorb any irrelevant or
			unacceptable transitions. Second, some strings $x=x_1\dots x_n$ have a different complexity from their reverse $x_n\dots x_1$.
			For instance \cite{COCOON:14,Kjos-EJC},
			\[
				A_D(011100) = 4 < 5 = A_D(001110).
			\]
			We tweak the definition of automatic complexity by introducing nondeterminism.
			\begin{df}[\cite{Kjos-EJC}]\label{precise}
				The nondeterministic automatic complexity $A_N(w)$ of a word $w$ is
				the minimum number of states of an NFA $M$
				(having no $\epsilon$-transitions) accepting $w$
				such that there is only one accepting path in $M$ of length $\abs{w}$.
			\end{df}
			\begin{figure}[h]
				\begin{tikzpicture}[shorten >=1pt,node distance=1.5cm,on grid,auto]
					\node[state,initial, accepting] (q_1)  {$q_1$}; 
					\node[state] (q_2)   [right=of q_1  ] {$q_2$}; 
					\node[state] (q_3)   [right=of q_2  ] {$q_3$}; 
					\node[state] (q_4)   [right=of q_3  ] {$q_4$};
					\node    (q_dots) [right=of q_4  ] {$\ldots$};
					\node[state] (q_m)   [right=of q_dots] {$q_m$};
					\node[state] (q_{m+1}) [right=of q_m  ] {$q_{m+1}$}; 
					\path[->] 
						(q_1)   edge [bend left] node      {$x_1$}   (q_2)
						(q_2)   edge [bend left] node      {$x_2$}   (q_3)
						(q_3)   edge [bend left] node      {$x_3$}   (q_4)
						(q_4)   edge [bend left] node [pos=.45] {$x_4$}   (q_dots)
						(q_dots) edge [bend left] node [pos=.6] {$x_{m-1}$} (q_m)
						(q_m)   edge [bend left] node [pos=.56] {$x_m$}   (q_{m+1})
						(q_{m+1}) edge [loop above] node      {$x_{m+1}$} ()
						(q_{m+1}) edge [bend left] node [pos=.45] {$x_{m+2}$} (q_m)
						(q_m)   edge [bend left] node [pos=.4] {$x_{m+3}$} (q_dots)
						(q_dots) edge [bend left] node [pos=.6] {$x_{n-3}$} (q_4)
						(q_4)   edge [bend left] node      {$x_{n-2}$} (q_3)
						(q_3)   edge [bend left] node      {$x_{n-1}$} (q_2)
						(q_2)   edge [bend left] node      {$x_n$}   (q_1);
				\end{tikzpicture}
				\caption{
					A nondeterministic finite automata that only accepts one string
					$x= x_1x_2x_3x_4 \ldots x_n$ of length $n = 2m + 1$.
				}
				\label{fig1}
			\end{figure}
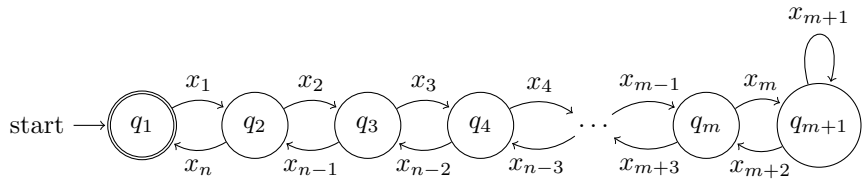
			Moreover, and most importantly for the present paper, $A_N$ gives rise to a striking instance of the idea of \emph{complexity deficiency}:
			\begin{thm}[\cite{Hyde,Kjos-EJC}]\label{Kayleigh}
				The nondeterministic automatic complexity $A_N(x)$ of a string $x$ of length $n$ satisfies
				\[
					A_N(x) \le b(n):={\lfloor} n/2 {\rfloor} + 1\text{.}
				\]
			\end{thm}
			\begin{proof}[Proof sketch]
				The proof is essentially contained in Figure \ref{fig1},
				although we must modify the picture slightly if $x$ has even length.
			\end{proof}
			\begin{df}
				The \emph{nondeterministic automatic complexity deficiency} of a string $x$ is defined by
				\[
					D_n(x) = b(n) - A_N(x),
				\]
				with $b(n)$ as in Theorem \ref{Kayleigh}. Sometimes we write $D(x)$ for $D_n(x)$.
			\end{df}
			Experimentally we have found that about half of all strings have $D_n(x)=0$ \cite{Kjos-EJC}.
			We call such strings \emph{complex}, and other strings \emph{simple}, herein.

		\subsection{Option types: perpetual, American, European}
			We shall consider the following types of options and their prices.
			\begin{itemize}
				\item[$V$.] This is the price of the perpetual option that pays out the deficiency $D_n(x)$ when
					we exercise the option at a time $n$.
					(\emph{Perpetual} here means that
					we can exercise the option at any time step labeled by a nonnegative integer.)
					The price of a perpetual option is the supremum, over all exercise policies ${\tau}$,
					of the expected payoff when using ${\tau}$. There is no restriction that ${\tau}$ be computable
					(in particular, there is no restriction that there be enough time to compute it before the next market time step occurs),
					but if that were to become an issue one would presumably change the definition accordingly.
				\item[$V_n$.] This is the price of an \emph{American} option that we can exercise at any time step
					labeled by an integer between $0$ and $n$.
				\item[$W_n$.] This is the price of the \emph{European} option with expiry $n$.
					In this case we must exercise the option at time $n$, if at all. So
					$W_n=\E(\max\{D_n, 0\})$.
					Here, and in the rest of this article, we assume the underlying probability distribution is given by the fair-coin measure.
					In a finance setting it could more generally be given by
					the risk-neutral measure determined from a stock price process.
			\end{itemize}
			We have
			\[
				\E D_n \le W_n \le V_n \le V,
			\]
			and
			\begin{thm}\label{three}
				\[
					\sup_n \E D_n \le \sup_n V_n \le V \le \E \sup_n D_n.
				\]
			\end{thm}
			\begin{proof}
				For the first inequality, it suffices to show
				\[
					\E D_n \le V
				\]
				for each $n$. This holds because one possible exercise policy is
				the static strategy of exercising at time $n$ no matter what.

				For the third inequality, there are two cases.

				Case 1: $\sup_n D_n$ is almost surely finite.
				Note that $D_n$ is integer-valued, so $\sup_n D_n$ will be realized at some finite stage $n_0$.
				Let us call \emph{magically prescient} the strategy which waits for
				$\sup_n D_n$ to be realized and then exercises the option.
				By contrast, an \emph{exercise policy} should be a stopping time, i.e.,
				it should not depend on future outcomes.
				We see that the payoff from the magically prescient strategy has
				a higher price than any exercise policy.
				It follows that $V\le\E\sup_n D_n$ in this case.

				Case 2: $\mathbb P(\sup_n D_n=\infty)>0$. Then $\E\sup_n D_n=\infty$ and so we are done.
			\end{proof}
			\begin{rem}
				In Case 2 of Theorem \ref{three}, if $\mathbb P(\sup_n D_n=\infty) = \eps > 0$ then we can even assert that $V=\infty$.
				Indeed if $V<\infty$ then we can buy the option, and wait for $D_n > V/\eps + 1$.
				The expected payoff is at least
				\[
					(\eps)(V/\eps + 1) = V + \eps > V,
				\]
				which would create an arbitrage.
			\end{rem}
			In Sections \ref{sec:KC} and \ref{sec:CFC} we shall consider several complexity notions, including
			\begin{itemize}
				\item prefix-free Kolmogorov complexity $K$,
				\item plain Kolmogorov complexity $C$, and
				\item nondeterministic automatic complexity $A_N$.
			\end{itemize}
			For each notion we first define one or more suitable deficiency notions $D_n(x)$:
			for instance, $D_n(x) = n + c_C - C(x)$ for a suitable constant $c_C$ for $C$, and
			$D_n(x) = \lfloor n/2\rfloor + 1 - A_N(x)$ for $A_N$.
			The following questions are natural for each of these deficiency notions:
			\begin{itemize}
				\item Does the price of the European option tend to $\infty$?
				\item Does the price of the American option tend to $\infty$?
				\item Does the American option have an efficiently computable exercise policy?
			\end{itemize}

	\section{Kolmogorov complexity}\label{sec:KC}
		\subsection{Plain complexity \texorpdfstring{$C$}{C}}
			Let $c_C$ be the least constant $c_C$ such that
			$C(x\mid n)\le n+c_C$ for all strings $x$ of any length $n$.
			If we define $D_n(x) = n + c_C - C(x\mid n)$ for $x$ of length $n$, then $D_n(x)\ge 0$ for all $x$,
			and $D_n(x)=0$ does occur. This is theoretically pleasant. Deficiencies are nonnegative and can
			be zero.
			Of course, $c_C$ depends on the version of the plain length-conditional Kolmogorov complexity $C(\cdot\mid\cdot)$ that we use.
			In this setting, we have
			\begin{thm}\label{plainFinite}
				$\sup_n\E D_n<\infty$.
			\end{thm}
			\begin{proof}
				Fix $n$. For any $a$, there are only $2^{a+1} - 1$ binary strings of length at most $a$.
				All descriptions witnessing complexity (given $n$) being at most $a$ must be among them, so
				at most $2^{a + 1} - 1$ many strings have complexity (given $n$) of at most $a$.
				(This is a standard argument, see \cite[Proposition 3.1.3]{MR2732288}.)
				Applying this to $a = n + c_C - k$,
				at most $2^{n + c_C - k + 1} - 1$ strings $x$
				(in particular, at most that many strings of length $n$) satisfy
				$D_n(x) \ge k$. That is,
				\[
					\mathbb P(D_n(x)\ge k) \le 2^{c_C - k + 1}.
				\]
				Then we have
				\[
					\E D_n = \sum_{k = 0}^\infty k \,\mathbb P(D_n = k)
					= \sum_{k = 1}^\infty \mathbb P(D_n \ge k)
					\le \sum_{k = 1}^\infty 2^{c_C - k + 1} = 2^{c_C + 1}.\qedhere
				\]
			\end{proof}
			It turns out that for options expiring at time $n$, there is a significantly better exercise policy than
			the static strategy of waiting until the very end:
			\begin{thm}\label{MLstrategy}
				For plain Kolmogorov complexity, $\sup_n V_n=\infty$,
				even if we require efficient computation of the exercise policy.
			\end{thm}
			The idea of the proof is to use complexity oscillations, first observed by \cite{MR0451322}:
			when the initial part of a string $x$ is a binary encoding of the length of $x$,
			the plain Kolmogorov complexity of $x$ will be low.
			\begin{proof}
					\cite{MR0451322} showed that deficiency is unbounded for all reals:
					for each $X$ and $b$ there is an $n$ with ${D({X \restrict n}) > b}$.
					We can computably identify such an $n$.
					The well known idea is that we take a prefix $X\restrict m$;
					consider it as a binary representation of a length $\ell < 2^m$;
					and then consider $\sigma=X\restrict\ell$.
					Since the beginning of $\sigma$ is known just from the length of $\sigma$,
					$\sigma$ is compressible.
					This translates into an exercise policy for our option:
					at the \emph{grant date} $m$ we decide on the date $\ell$ at which we are going to exercise.
Thus at the grant date our option style is transformed from American to European.
			\end{proof}

			\begin{rem}
				Since $C(x\mid n)\le^+ C(x)$, Theorem \ref{MLstrategy} holds equally for
				length-conditional plain Kolmogorov complexity, and
				Theorem \ref{plainFinite} also holds if we consider
				plain Kolmogorov complexity that is not length-conditional.
			\end{rem}
		\subsection{Prefix-free complexity \texorpdfstring{$K$}{K} with \texorpdfstring{$C$}{C}-style deficiency}
			Let $K$ denote prefix-free Kolmogorov complexity.
			With $D_n(x) = n - K(x)$, there is no limiting deficiency distribution in this case (or one could say the deficiency is in the limit $-\infty$ almost surely).
			That is, $K(w)\ge \abs{w}-c$ for almost all $w$, for any $c$.
			Indeed, for each $c\in\mathbb Z$,
			\[
				\lim_{n\rightarrow\infty}\frac{\left| \sigma\in 2^n: K(\sigma)\ge n-c\right|}{2^n} = 1,
			\]
			as is easily shown using $\sum_\sigma 2^{-K(\sigma)}<1$.
			If the $\limsup$ of the complement is $\delta>0$, then for each $\eps>0$ there exist $N_k$ with 
			\[
				1\ge \sum_\sigma 2^{-K(\sigma)} = \sum_n \sum_{\abs{\sigma}=n}2^{-K(\sigma)}
			\]
			\[	> \sum_k \delta(1-\eps) 2^{N_k}2^{-(N_k-c)}
				= (1-\eps)\delta \sum_k^\infty 2^c = \infty.
			\]
			\begin{thm}\label{prefixFinite}
				Let $K$ denote prefix-free Kolmogorov complexity $K$ and define the deficiency $D_n(x)=n-K(x)$
for a string $x$ of length $n$.
				The price of the perpetual option that pays $D_n - a$ is at most $2^{1 - a}$.
			\end{thm}
			\begin{proof}
				By \cite[Lemma 6.2.2]{MR2732288},
				\[
					\mathbb P(\sup_n D_n - a > c) = \mathbb P(\exists n\, K(X\restrict n) < n - c - a) \le 2^{-c - a}.
				\]
				Let ${D_n^+ = \max \{D_n - a, 0\}}$.
				Since we would not exercise an option giving negative payoff, it follows that
				\begin{eqnarray*}
					V \le \E(\sup_n D_n^+) &=& \sum_{c=0}^\infty c\, \mathbb P\left(\sup_n D_n^+ = c\right)
				\\
					= \sum_{c=1}^\infty c\, \mathbb P\left(\sup_n D_n^+ = c\right)
					&=& \sum_{c=1}^\infty \mathbb P\left(\sup_n D_n^+ \ge c\right)
				\\
					= \sum_{c=1}^\infty \mathbb P\left(\sup_n D_n - a > c-1\right) &\le& \sum_{c=1}^\infty 2^{-(c-1)-a} = 2^{1 - a}.\qedhere
				\end{eqnarray*}

			\end{proof}

		\subsection{Prefix-free complexity \texorpdfstring{$K$}{K} with its natural notion of deficiency}
			\begin{thm}[Deficiency based on an upper bound for \texorpdfstring{$K$}{K}]\label{Kbound}
				If we fix a constant $c_K$ such that for prefix-free Kolmogorov complexity $K$,
				\[
					K(x)\le n + K(n) + c_K
				\]
				for all $x$ of any length $n$,
				and let 
				\[ 
					D_n(x) = n + K(n) + c_K - K(x)\ge 0,
				\]
				then
				$\E D_n$ is bounded, but $V_n\rightarrow\infty$.
			\end{thm}
			\begin{proof}
				The same proof as for Theorem \ref{plainFinite} but using an analogous property
				shows that $\E D_n$ is bounded. In this case, however,
				$\sup D_n(X\restrict n)$ will be $\infty$ for almost all $X\in 2^\omega$.
				In fact Li and Vit\'anyi showed
				$D_n(X\restrict n) > \log n$ for infinitely many $n$ for almost all $X$.

				Solovay showed that	$\liminf D_n(X\restrict n)$ will be finite \cite{MillerYu}.

				$V = \infty$ in this case since we can simply wait for a sufficiently high $D_n$ value.
				What about $V_n$? Consider an arbitrary constant,
				which for expository vividness we will take to equal 17.
				Almost surely there will be an $n$ with $D_n(X\restrict n)\ge 17$. Therefore for each $\eps$
				there is an $n_0$ such that
				\[
					\mathbb P \bigcup_{n\le n_0} \{ D_n(X\restrict n)\ge 17 \} \ge 1 - \eps
				\]
				and so $V_{n_0}\ge 17(1 - \eps)$. Moreover $V_n\le V_{n + 1}$ for American options. 
				So $V_n\rightarrow\infty$ in this case.
				The exercise policy would be to wait for $D_n=17$ to occur and then exercise.
			\end{proof}

		\begin{sidewaystable}
			\centering
			\begin{tabular}{|c|c|c|c|c|}
				\hline
				$D_n$&
				$\sup_n \E D_n$&
				$\sup_n V_n$&
				$\E \sup_n D_n$\\
				\hline
				$n + c_K - K(x)$&
				$\therefore<\infty$&
				$\therefore<\infty$&
				$<\infty$ (Theorem\ \ref{prefixFinite})\\
				$n + K(n) + c_K - K(x)$&
				$<\infty$ (Theorem\ \ref{Kbound})&
				$\infty$ (Theorem\ \ref{Kbound})&
				$\therefore\infty$
				\\
				$n + c_C - C(x)$&
				$<\infty$ (Theorem\ \ref{plainFinite})&
				$\infty$ (Theorem\ \ref{MLstrategy})&
				$\therefore\infty$
				\\
				$n + c_C - C(x\mid n)$&
				$<\infty$ (Theorem\ \ref{plainFinite})&
				$\infty$ (Theorem\ \ref{MLstrategy})&
				$\therefore\infty$
				\\
				\hline
				$\lceil n/2\rceil + 1 - A_N(x)$&
				$<\infty$? (Conjecture \ref{oslo}) &
				$\infty$? (Conjecture \ref{oslo}) &
				$\therefore\infty$?
				\\
				\hline
			\end{tabular}
			\caption{Infinity and finiteness of option prices for various complexity deficiencies $D_n(x)$, for strings $x$ of length $n$.
			The conclusions labeled by $\therefore$ (``therefore'') follow from the inequalities
			$\sup_n \E D_n \le \sup_n V_n\le \E \sup_n D_n$
			(Theorem \ref{three}).
			}
			\label{priceOverview}
		\end{sidewaystable}

		An overview of the deficiency option prices is given in Table \ref{priceOverview}.
		\begin{rem}
			Of course, one does not need to only consider deficiencies.
			One could consider an option paying out $K(x) - n$.
			This value will go to infinity, but how fast?
			What is our exercise policy if we are not given access to $K$?
			Another possibility is to consider dips in complexity associated with the Kolmogorov structure function
			\cite{MR2103496}
			and its automatic complexity variant \cite{Kjos-TCS}.
		\end{rem}

		\subsection{Using runs}
			\begin{rem}\label{referee6}
				An anonymous referee suggested the following approach to obtaining results of the form $V_n\rightarrow\infty$.
				Let $R_n$ be the longest run of 0s in a string of length $n$ and let $\E$ and $\Var$ denote expectation and variance with respect
				to the uniform distribution on $\{0,1\}^n$.
				Now, if $U$ is a universal prefix-free machine,
				we can define another machine $M$ by the following algorithm: on input $x^*$, simulate $U$, and if $U(x^*)=x$, then
				\[
					M(x^*)=f(x):=x\,0^{\lfloor\log_2 \abs{x}\rfloor-c}
				\]
				for a fixed constant $c$.
				The domain of $M$ equals the domain of $U$, hence $M$ is also a prefix-free machine.
				Thus
				\[
					K(x\,0^{\lfloor\log_2 \abs{x}\rfloor-c})\le^+ K(x).
				\]
				Let now $m = \abs{f(x)} = \abs{x} + \lfloor\log_2 \abs{x}\rfloor - c$ and $y=f(x)$.
				Since $K(n)\le ^+ K(m)$ by the choice of $m$, we have
				\[
					K(y)\le n + K(n) + c_K \le^+ n + K(m) + c_K = (m + K(m) + c_K) - (m-n)
				\]
				and
				\[
					C(y)\le^+ n + c_C = m + c_C - (m - n).
				\]
				Now we employ the trading strategy whereby we wait until our input is of the form $x\,0^{\log_2 \abs{x}-c}$, and then exercise.
				By Theorem \ref{BoydThm} below, $\abs{\E(R_n)-\log_2 n}$ and $\Var(R_n)$ are both bounded by a constant $c$.
				By the argument in Section \ref{Malihe} below, with high probability we will be able to exercise.
				Thus for American options, with payoff $D_n(x)$ either $n + K(n) + c_K - K(x)$ or $c_C - C(x)$, we obtain $V_n\rightarrow\infty$.
			\end{rem}
			\begin{thm}[\cite{Boyd}]\label{BoydThm}
				Let $R_n$ be the longest run of heads in
				a binary sequence of length $n$
				distributed according to the Bernoulli distribution with parameter $1/2$. Let $\log=\ln$. Then
				\[
					\E(R_n)
					= \log_2 n + \frac{\gamma}{\ln 2} - \frac32 + \varepsilon_1(\log n/\log 2) + r_1(n),
				\]
				where $\varepsilon_1(\alpha)$ is a function of period 1 which satisfies
				$\abs{\varepsilon_1(\alpha)} < 2\times 10^{-6}$ for all $\alpha$, and
				$r_1(n) = O(n^{-1}(\log n)^4) \rightarrow 0$. 
				Moreover,
				\[
					\Var(R_n) = \frac1{12}+\frac{\pi^2}{6(\log 2)^2} + \eps_2\left(\frac{\log n}{\log 2}\right) + O(n^{-1}(\log n)^5),
				\]
				where $\eps_2(\alpha)$ has period 1, and $\abs{\eps_2(\alpha)}<10^{-4}$ for all $\alpha$.
			\end{thm}

	\section{Computable forms of complexity}\label{sec:CFC}
		\subsection{Automatic complexity}
			Now the goal is to price the European/American option that pays the
			nondeterministic automatic complexity deficiency $D_n$ of the movements of a stock
			from time 0 to the time $n$ when the option is exercised.
			We suspect that finding the exact price is a computationally intractable problem, both because of
			the conjectured intractability of computing automatic complexity \cite{Kjos-EJC}, and
			because of the exponential number of price paths to consider.
			
			The interest rate $r$ can be set to 0 or to a positive value.
			For pedagogical reasons, \cite{Shreve}
			uses $r=1/4$ for his main recurring example, and we sometimes adopt that value as well.
			\begin{itemize}
				\item For $n=0$ the option would pay 0 as there are no simple strings,
					and moreover the situation is anyway already known at time $0$.
				\item For $n=1$ the actual string (0 or 1) is not known at time 0 but
					it does not affect the payoff,
					which is 0 either way as there are no simple strings of length 1.
				\item For $n=2$, with up-factor $u=2$, down-factor $d=\frac12$, and $r=1/4$,
					there is a risk-neutral probability of $1/2$ of one of the strings $00$, $11$,
					both of which pay \$1. So the value is
					\[
						(1+r)^{-2}\cdot \frac12\cdot 1 = \frac{16}{50}.
					\]
			\end{itemize}
			In general when the risk-neutral probabilities are $1/2$ each for up and down,
			then the value of the option is directly related to the distribution of the deficiency $D_n$:
			\[
				\sum_{d=0}^{n/2} d \cdot \mathbb P(D_n=d) \cdot (1+r)^{-n} = \E(D_n)\cdot (1+r)^{-n}.
			\]
			If $D_n$ happened to be Poisson for large $n$, this is approximately $\lambda (1+r)^{-n}$,
			which is decreasing in $n$. However, we have just seen that
			the value for $n=2$ is higher than for $n=0$ and $n=1$.

			\begin{rem}\label{amirQ}
				For an American version, one question is
				whether to exercise the option at time $n=2$ after having seen $00$.
				If we exercise we get \$1. Otherwise the deficiency can at most go up by 1 each time step,
				whereas the interest factor with $r=1/4>0$ is exponential,
				so an upper bound for our payoff is
				\[
					(n/2)(1+r)^{-n} = \frac{n}{2}e^{-n\ln(5/4)}.
				\]
				This expression is maximized at $n=4$ and at $n=5$. Both places it takes the value $.8192$.
			\end{rem}

			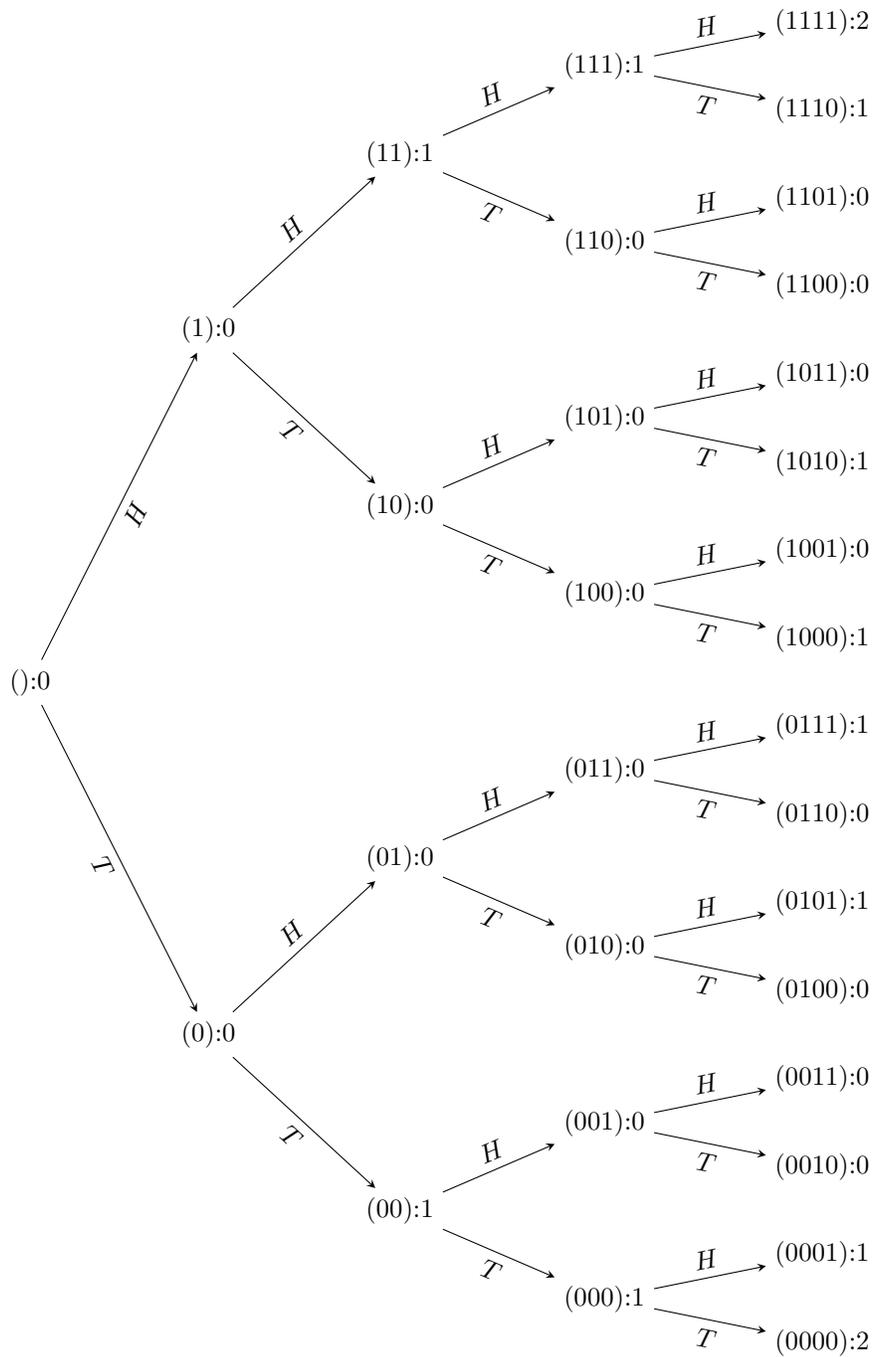
\begin{figure}
				\begin{tikzpicture}[>=stealth,sloped]
					\matrix (tree) [%
						matrix of nodes,
						column sep=1.5cm,
						row sep=0cm,
					]
					{
								&   	&		&   	& (1111):2\\
								&   	&		&(111):1&         \\
								&   	&		&   	& (1110):1\\
								&   	& (11):1&   	&         \\
								&   	&		&   	& (1101):0\\
								&   	&		&(110):0&         \\
								&   	&		&   	& (1100):0\\
								& (1):0 &		&   	&         \\
								&   	&		&   	& (1011):0\\
								&   	&		&(101):0&         \\
								&   	&		&   	& (1010):1\\
								&   	& (10):0&   	&         \\
								&   	&		&   	& (1001):0\\
								&   	&		&(100):0&         \\
								&   	&		&   	& (1000):1\\
						():0    &   	&		&   	&         \\
								&   	&		&   	& (0111):1\\
								&   	&		&(011):0&         \\
								&   	&		&   	& (0110):0\\
								&   	& (01):0&   	&         \\
								&   	&		&   	& (0101):1\\
								&   	&		&(010):0&         \\
								&   	&		&   	& (0100):0\\
								& (0):0 &		&   	&         \\
								&   	&		&   	& (0011):0\\
								&   	&		&(001):0&         \\
								&   	&		&   	& (0010):0\\
								&   	& (00):1&   	&         \\
								&   	&		&   	& (0001):1\\
								&   	&		&(000):1&         \\
								&   	&		&   	& (0000):2\\
					};
					\draw[->] (tree-2-4)  -- (tree-1-5)  node [midway,above] {$H$};
					\draw[->] (tree-2-4)  -- (tree-3-5)  node [midway,below] {$T$};
					\draw[->] (tree-4-3)  -- (tree-2-4)  node [midway,above] {$H$};
					\draw[->] (tree-4-3)  -- (tree-6-4)  node [midway,below] {$T$};
					\draw[->] (tree-6-4)  -- (tree-5-5)  node [midway,above] {$H$};
					\draw[->] (tree-6-4)  -- (tree-7-5)  node [midway,below] {$T$};
					\draw[->] (tree-8-2)  -- (tree-4-3)  node [midway,above] {$H$};
					\draw[->] (tree-8-2)  -- (tree-12-3) node [midway,below] {$T$};
					\draw[->] (tree-10-4) -- (tree-9-5)  node [midway,above] {$H$};
					\draw[->] (tree-10-4) -- (tree-11-5) node [midway,below] {$T$};
					\draw[->] (tree-12-3) -- (tree-10-4) node [midway,above] {$H$};
					\draw[->] (tree-12-3) -- (tree-14-4) node [midway,below] {$T$};
					\draw[->] (tree-14-4) -- (tree-13-5) node [midway,above] {$H$};
					\draw[->] (tree-14-4) -- (tree-15-5) node [midway,below] {$T$};
					\draw[->] (tree-16-1) -- (tree-8-2)  node [midway,below] {$H$};
					\draw[->] (tree-16-1) -- (tree-24-2) node [midway,below] {$T$};
					\draw[->] (tree-18-4) -- (tree-17-5) node [midway,above] {$H$};
					\draw[->] (tree-18-4) -- (tree-19-5) node [midway,below] {$T$};
					\draw[->] (tree-20-3) -- (tree-18-4) node [midway,above] {$H$};
					\draw[->] (tree-20-3) -- (tree-22-4) node [midway,below] {$T$};
					\draw[->] (tree-22-4) -- (tree-21-5) node [midway,above] {$H$};
					\draw[->] (tree-22-4) -- (tree-23-5) node [midway,below] {$T$};
					\draw[->] (tree-24-2) -- (tree-20-3) node [midway,above] {$H$};
					\draw[->] (tree-24-2) -- (tree-28-3) node [midway,below] {$T$};
					\draw[->] (tree-26-4) -- (tree-25-5) node [midway,above] {$H$};
					\draw[->] (tree-26-4) -- (tree-27-5) node [midway,below] {$T$};
					\draw[->] (tree-28-3) -- (tree-26-4) node [midway,above] {$H$};
					\draw[->] (tree-28-3) -- (tree-30-4) node [midway,below] {$T$};
					\draw[->] (tree-30-4) -- (tree-29-5) node [midway,above] {$H$};
					\draw[->] (tree-30-4) -- (tree-31-5) node [midway,below] {$T$};
				\end{tikzpicture}
				\caption{
					Deficiency tree for $n=4$, see Remark \ref{MFE-remark}.
				}\label{deficiency-tree}
			\end{figure}
			\begin{figure}
				\begin{tikzpicture}[>=stealth,sloped]
					\matrix (tree) [%
						matrix of nodes,
						minimum size=0.4cm,
						column sep=2.0cm,
						row sep=0.12cm,
						style={nodes={draw,rounded corners}},
					]
					{
								&   	&		&   & 2 \\
								&   	&		&1.2&   \\
								&   	&		&   & 1 \\
								&   	& 1     &   &   \\
								&   	&		&   & 0 \\
								&   	&		& 0 &   \\
								&   	&		&   & 0 \\
								& 0.528 &		&   &   \\
								&   	&		&   & 0 \\
								&   	&		&0.4&   \\
								&   	&		&   & 1 \\
								&   	& 0.32  &   &   \\
								&   	&		&   & 0 \\
								&   	&		&0.4&   \\
								&   	&		&   & 1 \\
						0.4224  &   	&		&   &   \\
								&   	&		&   & 1 \\
								&   	&		&0.4&   \\
								&   	&		&   & 0 \\
								&   	& 0.32  &   &   \\
								&   	&		&   & 1 \\
								&   	&		&0.4&   \\
								&   	&		&   & 0 \\
								& 0.528 &		&   &   \\
								&   	&		&   & 0 \\
								&   	&		& 0 &   \\
								&   	&		&   & 0 \\
								&   	& 1     &   &   \\
								&   	&		&   & 1 \\
								&   	&		&1.2&   \\
								&   	&		&   & 2 \\
					};
					\draw[->] (tree-2-4)  -- (tree-1-5)  node [midway,above] {$H$};
					\draw[->] (tree-2-4)  -- (tree-3-5)  node [midway,below] {$T$};
					\draw[->] (tree-4-3)  -- (tree-2-4)  node [midway,above] {$H$};
					\draw[->] (tree-4-3)  -- (tree-6-4)  node [midway,below] {$T$};
					\draw[->] (tree-6-4)  -- (tree-5-5)  node [midway,above] {$H$};
					\draw[->] (tree-6-4)  -- (tree-7-5)  node [midway,below] {$T$};
					\draw[->] (tree-8-2)  -- (tree-4-3)  node [midway,above] {$H$};
					\draw[->] (tree-8-2)  -- (tree-12-3) node [midway,below] {$T$};
					\draw[->] (tree-10-4) -- (tree-9-5)  node [midway,above] {$H$};
					\draw[->] (tree-10-4) -- (tree-11-5) node [midway,below] {$T$};
					\draw[->] (tree-12-3) -- (tree-10-4) node [midway,above] {$H$};
					\draw[->] (tree-12-3) -- (tree-14-4) node [midway,below] {$T$};
					\draw[->] (tree-14-4) -- (tree-13-5) node [midway,above] {$H$};
					\draw[->] (tree-14-4) -- (tree-15-5) node [midway,below] {$T$};
					\draw[->] (tree-16-1) -- (tree-8-2)  node [midway,below] {$H$};
					\draw[->] (tree-16-1) -- (tree-24-2) node [midway,below] {$T$};
					\draw[->] (tree-18-4) -- (tree-17-5) node [midway,above] {$H$};
					\draw[->] (tree-18-4) -- (tree-19-5) node [midway,below] {$T$};
					\draw[->] (tree-20-3) -- (tree-18-4) node [midway,above] {$H$};
					\draw[->] (tree-20-3) -- (tree-22-4) node [midway,below] {$T$};
					\draw[->] (tree-22-4) -- (tree-21-5) node [midway,above] {$H$};
					\draw[->] (tree-22-4) -- (tree-23-5) node [midway,below] {$T$};
					\draw[->] (tree-24-2) -- (tree-20-3) node [midway,above] {$H$};
					\draw[->] (tree-24-2) -- (tree-28-3) node [midway,below] {$T$};
					\draw[->] (tree-26-4) -- (tree-25-5) node [midway,above] {$H$};
					\draw[->] (tree-26-4) -- (tree-27-5) node [midway,below] {$T$};
					\draw[->] (tree-28-3) -- (tree-26-4) node [midway,above] {$H$};
					\draw[->] (tree-28-3) -- (tree-30-4) node [midway,below] {$T$};
					\draw[->] (tree-30-4) -- (tree-29-5) node [midway,above] {$H$};
					\draw[->] (tree-30-4) -- (tree-31-5) node [midway,below] {$T$};
				\end{tikzpicture}
				\caption{
					Option prices corresponding to Figure \ref{deficiency-tree}.
				}\label{option-price-tree}
			\end{figure}
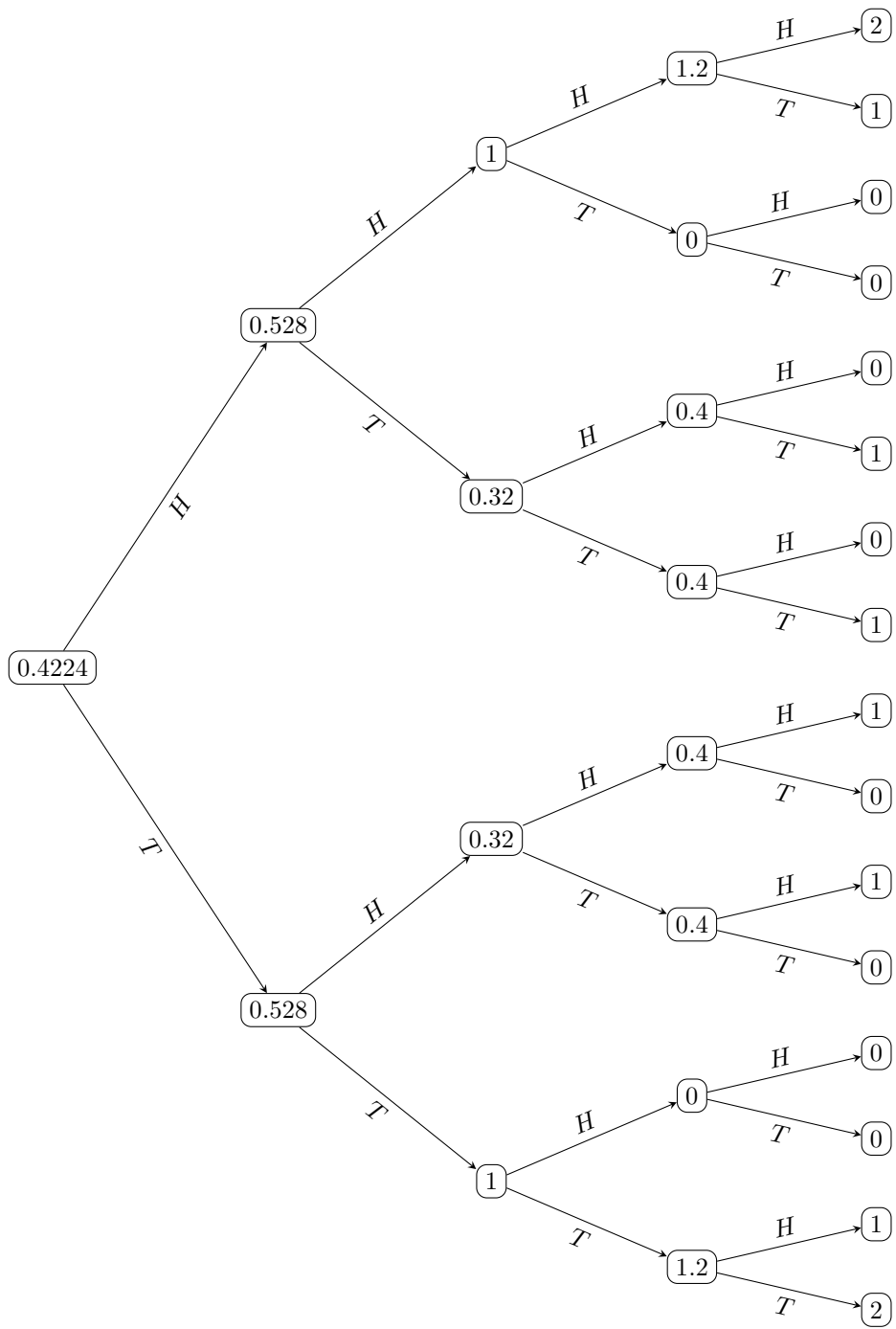

			\begin{table}
				\begin{center}
				\begin{tabular}{|r|l|c|l|}
					\hline
					Length&$\E D_n$& $\le$ & $V_n$\\
					\hline
					 0	&	$0$    	& $=$ &  	$0$\\
					 2	&	$0.5$  	& $=$ &  	$0.5$\\
					 4	&	$0.625$	& $<$ &  	$0.75$\\
					 6	&	$0.687$	& $<$ &  	$0.875$\\
					 8	&	$0.765$	& $<$ &  	$1.070$\\
					10	&	$0.791$	& $<$ &  	$1.191$\\
					12	&	$0.720$	& $<$ &  	$1.236$\\
					\hline
				\end{tabular}
			\end{center}
				\caption{
					Static versus dynamic exercise policies for nondeterministic automatic complexity.
				}
				\label{static}
			\end{table}

			To obtain a reasonable level of abstraction
			it is valuable to consider infinite price paths
			and associate a finite complexity deficiency with them. We can do so if
			the nondeterministic automatic complexity deficiencies of prefixes of an infinite binary sequence
			are almost surely bounded (Conjecture \ref{oslo}; see also Table \ref{priceOverview}).

			\begin{con}\label{oslo}
				For nondeterministic automatic complexity $A_N$,
				\[
					\mathbb P(\sup_n D_n<\infty) = 1,
				\text{ and yet}\quad
					\sup_n V_n = \infty.
				\]
			\end{con}

			\begin{rem}\label{MFE-remark}
				\cite{MFE} studied
				a perpetual American option that pays the complexity deficiency of the sequence of up and down ticks
				(considered as 1s and 0s) upon exercise.
				With interest rate set to zero ($r=0$) the price of this security may be infinity, based on tentative numerical evidence.
				That is, for $A_N$,
				\[
					\sup_n V_n = \infty,
				\]
				although $\E D_n$ seems to approach a finite limit (see Table \ref{static}).
				For positive interest rates the price is finite (see Remark \ref{amirQ}).
				They found numerical evidence that for $r=1/4$ the price is $0.47$.
				See Figure \ref{deficiency-tree} for the deficiencies of strings of length at most 4,
				and Figure \ref{option-price-tree} for corresponding calculated option prices.
				The price of the American option with expiry $2k$ and expiry $2k+1$ are the same,
				as is easy to prove.
			\end{rem}

			\begin{df}
				Let $W_n$ be the price of the European option paying the
				nondeterministic automatic complexity deficiency $D(x)$ for the price path $x$ of length $n$.

				\noindent\emph{Decision problem:} $\mathrm{PRICE}$.

				\noindent\emph{Instance:} A pair of nonnegative integers $n$ and $k$ with
				\[
					0\le \frac{k}{2^n}\le\lfloor n/2\rfloor +1.
				\]

				\noindent\emph{Question:} Is $W_n\ge k/2^n$?
			\end{df}
			Recall that $\mathrm{E}$ is the class of single-exponential time decidable decision problems.
			\begin{thm}\label{referee}
				$\mathrm{PRICE}$ is in $\mathrm{E}$.
			\end{thm}
			\begin{proof}
				\cite{Kjos-EJC} considered the problem $\mathrm{DEFICIENCY}$ of deciding whether,
				given an integer $k$ and a sequence $x$,
				the nondeterministic automatic complexity deficiency $D(x)$ satisfies $D(x)\ge k$.
				They showed that $\mathrm{DEFICIENCY}$ is in $\mathrm{E}$.
				Since there are only single-exponentially many price paths of length $n$,
				the usual backwards recursive algorithm for
				option pricing in the binomial model \cite{Shreve} gives the theorem.
			\end{proof}
			The same proof shows that
			the analogous statement to Theorem \ref{referee} for American options holds as well.

		\subsection{Run complexity}\label{Malihe}
			If the payoff of our option is just the longest run of heads then \cite{Alikhani} showed that
			the price of the option is $\Theta(\log_2 n)$.
			This corresponds to automata that always proceed to a fresh state,
			except that one state may be repeated (namely, the state of the longest run).

			\begin{df}\label{df:runComplexity}
				The \emph{run complexity} $C_R$ of a binary sequence $x$ is defined by $C_R(x)=n+1-r$,
				where $n$ is the length of $x$ and $r$ is the length of the longest run of 0s or 1s in $x$.
			\end{df}
			This complexity notion has the advantage that it is efficiently computable.
			\cite{Kjos-COCOA} studied it in more detail
			and also considered multiple runs, as in the Wald--Wolfowitz runs test.

			In the rest of this subsection we give the argument of \cite{Alikhani}.
			We assume familiarity with basic discrete options \cite{Shreve}.
			A coin tossing sequence is $\omega_1\dots\omega_N$ where each $\omega_i\in\{H,T\}$.
			(Read $H$ as ``heads'' and $T$ as ``tails''.)
			\begin{df}
				For each $0\le n\le N$, the \emph{current run of heads} in the coin tossing sequence $\omega_1\dots\omega_n$ is defined by
				\[
					G_n(\omega)=\max \{r: \omega_{n-r+1} = \dots = \omega_n = H\}.
				\]
				The \emph{run option} is the American option where the payoff when exercised at time $n\le N$ is
				$G_n(\omega)$.
				Let $V^A $ be the price of the run option.
				Define a stopping time $\tau_t$ by
				\[
					\tau_t(\omega_1 \dots \omega_N)=\min\{s: G_s=[\E(R_N)] - t\},
				\]
				where $R_N$ is the longest run of heads in a coin tossing sequence of length $N$.
			\end{df}
			Thus, the trading strategy corresponding to $\tau_t$ is to wait for a run of heads
			that is almost as long as we ever expect to see before time $N$,
			with ``almost'' being qualified and measured by the parameter $t$.
			\begin{df}
				Let $[x]$ denote the nearest integer of $x$. In particular,
				$[x]$ is an integer $k$ with $k-1\le x\le k+1$.
			\end{df}				

			\begin{thm}\label{bumblebee}
				Given $N$ there is a deterministic choice of $t=t_N$ such that
				there is a sequence of numbers $\varepsilon_N$ with
				$\lim_{N \to +\infty} \varepsilon_N=0$, and
				constants $c_2$ and $c$, such that for large $N$,
				\[
					\E(G_{\tau_{t_{N}}})\geq (\log_2N-c_2- c \sqrt[3] {\ln N})(1-\varepsilon_N).
				\]
			\end{thm}
			\begin{proof}
				Let $S_n$ be the set of stopping times taking values in $\{n, n+1,\dots,\infty\}$
				\cite[Section 4.4]{Shreve}.
				The price process $V_n^A$ for the run option satisfies the American risk-neutral formula
				\[
					V_n^A=\max_{\tau \in S_n} {\E_n}[\mathbb{I}_{\tau \leq N} G_{\tau}],
					\quad \text{for } n=0,1,\dots,N.
				\] 
				So for each $t$,
				\[
					V_n^A \geq {\E_n}[(\mathbb{I}_{\tau_t \leq N}) G_{\tau_t}].
				\]
				Now we find a lower bound on $\E(G_{\tau_t})$.
				\begin{align*}
					\E(G_{\tau_{t}})
					&=  \E(G_{\tau_{t}}| G_{\tau_{t}} > 0) \Pr(G_{\tau_{t}} > 0)\\
					&=  ([\E(R_N)] - t)(\Pr\{R_N \ge [\E(R_N)] - t  \})\\
					&\ge ([\E(R_N)] - t)(\Pr\{R_N \ge \E(R_N) - t + 1 \})\\
					&\ge ([\E(R_N)] - t)(\Pr\{|R_N- \E(R_N)| \leq (t-1)\})\\
					&\ge ([\E(R_N)] - t)  \left(1-\frac{\Var(R_N)}{(t-1)^2}\right)
					\qquad\text{(by Chebyshev's Inequality)}\\
					&\ge (\E(R_N) - 1 - t)\left(1-\frac{\Var(R_N)}{(t-1)^2}\right).
				\end{align*}
				By Theorem \ref{BoydThm},
				\[
					\Var(R_N) = \pi ^2/6\ln^2(2)+1/12+r_2(N)+\varepsilon_2(N) \leq 4
				\]
				for large $N$.
				Let $\E(R_N)=a$; then we get
				\begin{equation}\label{star}
					\E(G_{\tau_t})\ge (a-t-1)\left(1-\frac4{(t-1)^2}\right).
				\end{equation}
				Now we find the $t=t_N$ such that
				the right-hand side of (\ref{star}) is maximized.
				The corresponding third degree polynomial has negative discriminant.
				Therefore it has one real root, which was calculated by $\mathrm{Mathematica}$:
				\[
					t=\frac{
						\left(\frac{2}{3}\right)^{2/3}\sqrt[3]{
							9\ln^2(2)\ln (N)+\sqrt{3} \sqrt{27 \ln^4(2) \ln^2(N)+4 \ln^6(2)}
						}
					}{\ln (2)}
				\]
				\[
					-\frac{2\sqrt[3]{\frac{2}{3}} \ln(2)}{
						\sqrt[3]{9 \ln^2(2) \ln(N)+\sqrt{3} \sqrt{27 \ln ^4(2) \log^2(N)+4 \ln ^6(2)}}
					}.
				\]
				By the second derivative test, since
				\[
					\frac{d^2}{dt^2} (a-t)(1-4/t^2) = -3t^2-4 \leq 0,
				\]
				we see that $t$ maximizes the right-hand side of ($*$).
				We have
				\[
					\lim_{n\to\infty}-\frac{2\sqrt[3]{\frac{2}{3}} \ln (2)}
					{\sqrt[3]{9 \ln ^2(2) \ln (n)+\sqrt{3} \sqrt{27 \ln ^4(2) \ln ^2(N)+4 \ln ^6(2)}}} = 0
				\]
				and hence
				\[
					\frac{t}{(4/\ln 2)^{1/3}( \sqrt[3] {\ln N})}\rightarrow 1.
				\]
				Therefore, $t = t_N \in \Theta ( \sqrt[3] {\ln N})$ and so
				\[
					\E(G_{\tau_{t_{N}}}) \geq (\log_2N-c_2- c\sqrt[3] {\ln N})(1-\varepsilon_N).
				\]
			\end{proof}
			\begin{cor}\label{th1}
				$V^{A} \sim \log_2N$.
			\end{cor}
			\begin{proof}
				$V^{A} $ is bounded below by the expected payoff of the strategy that
				waits for $\mathbb [E(R_N)]-t_N$ heads, with $t_N$ as in Theorem \ref{bumblebee}, and then exercises.
				On the other hand, $V^{A} $ is bounded above $\E(R_N)$.
				Therefore
				\[
					\E(R_N)-t_N\leq V^{A} \leq \E(R_N).
				\]
				By Theorem \ref{BoydThm},
				\[
					\E(R_N)=\log_2(N/2)+ \gamma /\ln2-1/2=\log_2N+O(1),
				\]
				and so by Theorem \ref{bumblebee},
				\[
					\log_2N -c_2 -\sqrt[3] {\ln N} \leq V^{A} \leq \log_2N+O(1).
				\]
				Dividing by $\log_2 N$ we get
				\[
					1- o(1) \leq \frac{V^{A}}{\log_2 N} \leq 1+ o(1).\qedhere
				\]
			\end{proof}
	\section{Robustness}
		We now consider whether, in the phrase of an anonymous referee,
		\begin{quote}
			small perturbations on input sequences can have drastic effects on
			our studied measurements of complexity.
		\end{quote}
		In other words, whether errors in the measurement of a sequence
		will lead to large errors in the calculated complexity.
		Let $d(x,y)$ denote the Hamming distance between two sequences of the same length $x$ and $y$.
		Let us consider our three types of complexity in turn.
		\paragraph{Run complexity.}
			Here a change in a single bit sometimes cuts the longest run in half.
			That is,
			if $d(x,y)=1$ then $C_R(x)=n-r_x$ and $C_R(y)=n-r_y$ where $r_x\le 2 r_y + 1$.

			On the other hand, since
			the longest run will only be about $\log_2 n$ \cite{Boyd},
			a random change in a single random bit will tend to leave the complexity unchanged.
		\paragraph{Automatic complexity.}
			Here we have numerical evidence that a change in a single bit sometimes has large effects.
			For instance, consider the string $0^n$ which becomes $0^a10^{n-a-1}$. See Table \ref{perturbing}.
		\paragraph{Kolmogorov complexity.}
			A change in a single bit will affect the complexity only logarithmically
			(by at most about $2\log n$)
			since a description of the sequence can include
			hard-coded information about where the changed bit is.
			\cite{FLV} studied Kolmogorov complexity with error in detail.
			\begin{table}
				\centering
				\begin{tabular}{|c|c|}
					\hline
					$w$&$A_N(w)$\\
					\hline
					$0^{23}$   & 1\\
					$0^{22} 1$  & 2\\
					$0^{21} 1 0$ & 3\\
					$0^{20} 1 0^2$& 4\\
					$0^{19} 1 0^3$& 5\\
					$0^{18} 1 0^4$& 6\\
					$0^{17} 1 0^5$& 7\\
					$0^{16} 1 0^6$& 8\\
					$0^{15} 1 0^7$& 9\\
					$0^{14} 1 0^8$& 8\\
					$0^{13} 1 0^9$& 8\\
					$0^{12} 1 0^{10}$&8\\
					$0^{11} 1 0^{11}$&7\\
					\hline
				\end{tabular}
				\caption{
					Nondeterministic automatic complexity in the Hamming ball of radius 1 around $0^{n}$, $n=23$.
				}\label{perturbing}
			\end{table}
	\section{Enhanced Content}
		\paragraph{AutoComplex.} This app for Android devices \cite{links/AC}
		lets you look up nondeterministic automatic complexity values of particular strings.
		The app tells you the complexity of a given string and also provides a ``proof'' or ``witness''.

		This witness is a uniquely accepting state sequence,
		i.e., a sequence of states visited during a run of a witnessing automaton.
		It is analogous to a shortest description $x^*$ of a string $x$,
		familiar from the study of Kolmogorov complexity.

		The app also provides some extensions of the string suggested by
		the familiar autocompletion feature used in search engines.

		\paragraph{The Complexity Guessing Game and the Complexity Option Game.} These two online games \cite{links/CGG,links/COG} invite the player to guess complexities,
		or implement an exercise policy for a complexity-based financial option, respectively.
		The games include graphical displays of millions of the relevant automata.
	\section*{Acknowledgments}
		This work was partially supported by
		a grant from the Simons Foundation
		(\#315188 to Bj\o rn Kjos-Hanssen).

		This material is based upon work supported by the National Science Foundation under Grant No.\ 0901020.
	\bibliographystyle{alpha}
	\bibliography{randomness-risk}
\end{document}